\documentclass[10pt,a4paper]{article}
\usepackage{amssymb}
\usepackage{amsmath}
\usepackage{amsfonts}
\usepackage{amsthm}
\usepackage{latexsym}
\usepackage{mathrsfs}
\usepackage{stmaryrd}
\usepackage[dvips]{epsfig}
\usepackage{setspace}
\usepackage{float}
\usepackage{bm}
\usepackage{tikz}
\usepackage{enumerate}
\usepackage{subfigure}
\usepackage{nomencl}
\usepackage{authblk}

\usepackage{textcomp}
\usepackage{scrextend}
\usepackage[toc,page]{appendix}

\theoremstyle{plain}
\newtheorem{proposition}{Proposition}
\newtheorem{lemma}{Lemma}
\newtheorem{theorem}{Theorem}
\newtheorem{assumption}{Assumption}

\newtheorem{definition}{Definition}
\newtheorem{remark}{Remark}

\def\bmh{{\bm h}}

\def\bmdelta{{\bm \delta}}

\newcounter{mnotecount}

\newcommand{\mnotex}[1]
{\protect{\stepcounter{mnotecount}}$^{\mbox{\footnotesize $\bullet$\themnotecount}}$ 
\marginpar{
\raggedright\tiny\em
$\!\!\!\!\!\!\,\bullet$\themnotecount: #1} }

\begin{document}


\title{\textbf{Dain's invariant on non-time symmetric initial data sets}}

\author[,1]{J.A. Valiente Kroon \footnote{E-mail address:{\tt
      j.a.valiente-kroon@qmul.ac.uk}}}
\author[,1]{J.L. Williams \footnote{E-mail address:{\tt j.l.williams@qmul.ac.uk}}}
\affil[1]{School of Mathematical Sciences, Queen Mary, University of London,
Mile End Road, London E1 4NS, United Kingdom.}

\maketitle
\begin{abstract}
We extend Dain's construction of a geometric invariant characterising
static initial data sets for the vacuum Einstein field equations to
situations with a non-vanishing
extrinsic curvature. This invariant gives a measure of how much an
initial data set with non-vanishing ADM 4-momentum deviates from stationarity. In
particular, it vanishes if and only if the initial data set is
stationary. Thus, the invariant provides a quantification of the
amount of gravitational radiation contained in the initial data set. 
\end{abstract}


\section{Introduction}

The Cauchy problem for the Einstein field equations is a cornerstone
of  mathematical Relativity. Indeed, the proper rigorous formulation of many of the 
outstanding problems in mathematical Relativity, such as stability of
certain special solutions,  is made
within the framework of the Cauchy problem. Accordingly, one of the challenges in
the construction of the spacetimes by means of the Cauchy problem is
to provide physically relevant initial data for the evolution
equations of General Relativity. 

As it is well know, initial data for the evolution equations of
General Relativity cannot be freely specified ---in order to obtain a
proper solution to the Einstein field equations, the initial data set
has to satisfy the so-called Einstein constraint equations. Solutions
to the Einstein constraint equations have been studied extensively
---see e.g. \cite{BarIse04}. Of particular relevance in this respect
is the question of under which conditions an initial data set gives
rise to a spacetime development possessing Killing symmetries ---this
question first arose in the context of linearisation stability, see
\cite{Mon75}. These conditions are encoded in the so-called
\emph{Killing initial data (KID) equations} ---see
e.g. \cite{Chr91b,BeiChr97b} for a discussion of the basic properties
of these equations; see also \cite{Col77}. The KID equations
constitute a system of overdetermined equations for a scalar and a
vector on the initial hypersurface. The existence of a solution to
these equations is equivalent to the existence of a Killing vector in
the development of the initial data set. The KID equations have a deep
connection with the \emph{Arnowit-Deser-Misner (ADM) evolution
equations}: the evolution equations can be described as a flow
generated by the adjoint linearised constraint map, $D\Phi^*$ (see
below) ---see e.g. \cite{FisMar72} for further details.

In many applications of both physical and mathematical interest it is
important to have a way of quantifying how much a give initial data
set deviates from stationarity. Ideally, one would like to do this in
coordinate-independent manner. One approach to this problem was
proposed in \cite{Dai04c}, in which the notion of an \textit{approximate Killing vector}, as a solution to a fourth-order
linear elliptic system arising from the KID equations, was
introduced. The so-called \textit{approximate Killing vector equation}
has the property that its kernel contains that of the KID
equations. The analysis in \cite{Dai04c} was restricted to the case of
time symmetric asymptotically Euclidean initial data sets. In
particular, it was shown that the kernel of the approximate Killing
vector equation is non-trivial, and moreover that, 
given suitable assumptions on the asymptotics of the initial data set,
the solution (termed the \textit{approximate Killing vector}) is
unique up to constant rescalings.

It is of interest to mention that the general strategy adopted in
\cite{Dai04c} has found applicability in the analysis of spacetimes
admitting a Killing spinor ---see \cite{GarVal08a}. These ideas have
been used, in turn, to obtain an invariant characterising 
initial data sets for the Kerr spacetime, see
\cite{BaeVal10a,BaeVal10b}, and for the Kerr-Newman spacetime, see
\cite{ColVal16b}. 

\medskip
The purpose of this article is to extend Dain's result in
\cite{Dai04c} to the non-time symmetric case. Moreover, we analyse in
some detail conformally flat initial data sets as way of obtaining some
further insight into Dain's construction. Our main result is Theorem
\ref{MainResult} which shows that the approximate Killing
vector equation  can be solved with the required asymptotic conditions
for a large class of asymptotically Euclidean initial data sets. 

\subsection*{Overview of the article}
This article is structured as follows: Section
\ref{Section:ApproximateKillingVectorEquation} provides a discussion
of the basic properties of the approximate Killing vector equation as
introduced by Dain. In particular, Subsection
\ref{Subsection:KIDEquations} provides a discussion of the relation
between the Einstein constraint equations and the so-called Killing
Initial Data (KID) equations; Subsection
\ref{Subsection:BasicProperties} provides a detailed discussion of the
approximate Killing vector equation in the non-time symmetric setting;
Subsection \ref{Subsection:IntegrationParts} introduces some useful
identities which will be used throughout. Section
\ref{Section:AsymptoticallyEuclidean} analyses the solvability of the
approximate Killing equation on asymptotically Euclidean manifolds: in
Subsection \ref{Section:SobolevSpaces} some basic background on
weighted Sobolev spaces is given; Subsection
\ref{Section:DecayAssumptions} provides a discussion of our main
asymptotic decay assumptions and of the asymptotic behaviour of
solutions to the KID equations; Subsection \ref{Section:Ellipticity}
briefly reviews the basic methods to analyse the existence of
solutions to elliptic equations on asymptotically Euclidean manifolds;
Subsection \ref{Section:ExistenceSolutions} contains our main
existence results. Finally, Section \ref{Section:FurtherAnalysis} contains a further discussion of
the geometric invariant obtained from Dain's construction with
particular emphasis to the case of conformally flat initial data sets.

\subsection*{Notation and conventions}
We use Penrose's abstract index notation throughout so that ${}_i, \,
{}_j, \, {}_k, \ldots$ denote abstract 3-dimensional tensorial
indices. The Greek indices ${}_\alpha, \,{}_\beta, \, {}_\gamma,
\ldots$ denote 3-dimensional coordinate indices. Riemannian
3-dimensional metrics are assumed to have signature $(+++)$. Our
conventions for the curvature are fixed by 
\[
D_i D_j v^k -D_j D_i v^k = r^k{}_{lij} v^l.
\]
The Ricci $r_{lj}$ tensor is obtained from the Riemann tensor via the
relation
\[
r_{lj} \equiv r^i{}_{lij}.
\]

\section{The approximate Killing vector equation}
\label{Section:ApproximateKillingVectorEquation}

In this section  we introduce the basic objects of our analysis: the
vacuum Einstein constraint equations, the Killing initial data
equations and the approximate Killing initial data equations. 

\subsection{The Einstein constraints and the KID equations}
\label{Subsection:KIDEquations}
In this article we will study properties of initial data sets for the vacuum
Einstein field equations ---that is, triples
$(\mathcal{S},h_{ij},K_{ij})$ where $\mathcal{S}$ is a 3-dimensional
manifold, $h_{ij}$ is a Riemannian metric on $\mathcal{S}$ and $K_{ij}$
is a symmetric rank 2 tensor satisfying the vacuum Einstein constraint
equations
\begin{subequations}
\begin{eqnarray}
&& r+ K^2 -K_{ij} K^{ij} =0, \label{HamiltonianConstraint}\\
&& D^j K_{ij} - D_i K=0. \label{MomentumConstraint}
\end{eqnarray}
\end{subequations}
Following the standard conventions we refer to equations
\eqref{HamiltonianConstraint} and \eqref{MomentumConstraint} as the
\emph{Hamiltonian} and \emph{momentum} constraints, respectively. In
the above expressions $D_i$ denotes the Levi-Civita connection of the
metric $h_{ij}$ and $r$ is the associated Ricci scalar. Furthermore,
$K\equiv K_{ij} h^{ij}$. 

\medskip
In the following we will be particularly interested on initial data
sets $(\mathcal{S},h_{ij},K_{ij})$ whose development has a Killing
vector. The conditions for this to be case are identified in the following:

\begin{proposition}
Let $(\mathcal{S},h_{ij},K_{ij})$ denote an initial data set for the
vacuum Einstein field equations. If there exists and scalar field $N$ and a
vector field $Y^i$ over $\mathcal{S}$ satisfying the equations
\begin{subequations}
\begin{eqnarray}
&& L_{ij}\equiv N K_{ij} + D_{(i} Y_{j)}=0, \label{KID1}\\
&& M_{ij}\equiv Y^k D_k K_{ij} + D_i Y^k K_{kj} + D_j Y^k K_{ik} + D_i
D_j N \nonumber \\
&& \hspace{4cm} -N (r_{ij} + K K_{ij} -2 K_{ik} K^k{}_j)=0, \label{KID2}
\end{eqnarray}
\end{subequations}
then the development of the initial data is endowed with a Killing vector.
\end{proposition}

A proof of this result can be found in e.g. \cite{Chr91b} ---see also
\cite{BeiChr97b}.   

\begin{remark}
{\em The pair $(N,Y^i)$ is called a \emph{Killing initial data set (KID)} and
equations \eqref{KID1}-\eqref{KID2} are known as the \emph{KID equations}.}
\end{remark}

It is interesting to note that Killing initial data for conformally rescaled 
vacuum spacetimes has been analysed
in \cite{Paetz14,Paetz16}, with applications to the characterisation of 
Kerr-de Sitter-like spacetimes in \cite{Mars16}.

\subsection{Basic properties of the approximate Killing vector
  equation}
\label{Subsection:BasicProperties}
In the following, denote by $\mathscr{M}_2$, $\mathscr{S}_2$,
$\mathscr{X}$ and $\mathscr{C}$ the spaces of Riemannian metrics,
symmetric 2-tensors, vectors and scalar functions on the 3-dimensional
manifold $\mathcal{S}$, respectively. It is convenient to write the Einstein constraint
equations \eqref{HamiltonianConstraint} and \eqref{MomentumConstraint}
in terms of a map (\emph{the constraint operator})
\[
\Phi: \mathscr{M}_2\times \mathscr{S}_2 \rightarrow \mathscr{C}\times \mathscr{X}
\]
such that for  $h_{ij}\in \mathscr{M}_2$, $K_{ij}\in \mathscr{S}_2$
one has
\[
\Phi\left(\begin{array}{c}
 h_{ij}\\
 K_{ij}
\end{array}\right)\equiv\left(\begin{array}{c}
r+K^2-K_{ij}K^{ij}\\
-D^jK_{ij}+D_iK
\end{array}\right).
\]
In terms of the latter, the constraints  \eqref{HamiltonianConstraint}
and \eqref{MomentumConstraint} take the form 
\[
\Phi \left(\begin{array}{c}
 h_{ij}\\
 K_{ij}
\end{array}\right)=0.
\]
The \emph{linearisation of the constraint operator} $\Phi$, $D\Phi:
\mathscr{S}_2\times \mathscr{S}_2\rightarrow \mathscr{C}\times \mathscr{X}$,
evaluated at $(h_{ij},K_{ij})$ can be found to be given by
\[
D\Phi\left(\begin{array}{c}
\gamma_{ij}\\
Q_{ij}
\end{array}\right)=\left(\begin{array}{c}
D^iD^j\gamma_{ij}-r_{ij}\gamma^{ij}-\Delta_\bmh \gamma+H\\
-D^jQ_{ij}+D_iQ-F_i
\end{array}\right)
\]
where $\gamma\equiv h^{ij}\gamma_{ij}$, $Q\equiv h^{ij}Q_{ij}$ and 
\begin{align*}
& H\equiv 2(KQ-K^{ij}Q_{ij})+2(K^{ki}K^j{}_k-KK^{ij})\gamma_{ij},\\
& F_i\equiv \left(D_i K^{kj}-D^k K^j{}_i\right)\gamma_{jk}-\left(K^k{}_i D^j-\tfrac{1}{2}K^{kj}D_i\right)\gamma_{jk}+\tfrac{1}{2}K^k{}_i D_k \gamma,
\end{align*}
while $\Delta_\bmh\equiv h^{ij}D_i D_j$ is the Laplacian of the metric
$h_{ij}$. Moreover, using integration by parts, the formal adjoint of
 the linearised constraint operator,
 $D\Phi^*:\mathscr{C}\times\mathscr{X}\rightarrow \mathscr{S}_2\times \mathscr{S}_2$, can be seen to be given by 
\[
D\Phi^*\left(\begin{array}{c}
X\\
X_i
\end{array}\right)=\left(\begin{array}{c}
D_iD_jX-X r_{ij}-\Delta_h X h_{ij}+H_{ij}\\
D_{(i}X_{j)}-D^kX_k h_{ij}+F_{ij}
\end{array}\right)\]
where 
\begin{eqnarray*}
&& H_{ij}\equiv
  2X(K^k{}_iK_{jk}-KK_{ij})-K_{k(i}D_{j)}X^k+\tfrac{1}{2}K_{ij}D_kX^k\\
&& \hspace{2cm}+\tfrac{1}{2}K_{kl}D^kX^l h_{ij}-\tfrac{1}{2}X^kD_k K_{ij}+\tfrac{1}{2}X^k D_k K h_{ij},\\
&& F_{ij}\equiv 2X(Kh_{ij}-K_{ij}).
\end{eqnarray*}
Note that in the case of time-symmetric data, $H=F_i=H_{ij}=F_{ij}\equiv 0$, and the above expressions for $D\Phi$ and $D\Phi^*$ thereby reduce to those given in \cite{Dai04c}.

\begin{remark}
\label{Rmk:EquivalenceOfKIDsandDPhiStar}
{\em A calculation shows that $D\Phi^*=0$ is equivalent to the KID equations
\eqref{KID1}-\eqref{KID2}. Indeed, one has that 
\[D\Phi^*\left(\begin{array}{c}
N\\
-2Y_i
\end{array}\right)=\left(\begin{array}{c}
M_{ij}-M_k{}^k h_{ij}-\tfrac{1}{2}K_{kl}L^{kl} h_{ij}+\tfrac{1}{2}K_{ij}L_k{}^k\\
L_{ij}-L_k{}^k h_{ij}
\end{array}\right)
\]
from which we see that $D\Phi^*(N,-2Y_i)=0$ if and only if $L_{ij}=M_{ij}=0$ ---i.e. if and only if $(N,Y^i)$ satisfy the KID equations.}
\end{remark}

\medskip
Now, let $\mathscr{S}_{1,2}$ denote
the space of covariant rank-3 tensors which are symmetric in the last
two indices. Following Dain \cite{Dai04c}, we consider an operator
$\mathcal{P}:\mathscr{S}_2\times \mathscr{S}_{1,2}\rightarrow
\mathscr{C}\times \mathscr{X}$ such that 
\[
\mathcal{P}\left(\begin{array}{c}
\gamma_{ij}\\
q_{kij}
\end{array}\right)\equiv D\Phi\left(\begin{array}{c}
\gamma_{ij}\\
-D^kq_{kij}
\end{array}\right)
\]
with formal adjoint, $\mathcal{P}^*:\mathscr{C}\times
\mathscr{X}\rightarrow \mathscr{S}_2\times \mathscr{S}_{1,2}$, given by
\[
\mathcal{P}^*\left(\begin{array}{c}
X\\
X_i
\end{array}\right)\equiv \left(\begin{array}{cc}
1 & 0 \\
0 & D_k
\end{array}\right)\cdot D\Phi^*\left(\begin{array}{c}
X\\
X_i
\end{array}\right)=\left(\begin{array}{c}
D_iD_jX-X r_{ij}-\Delta_h X h_{ij}+H_{ij}\\
D_k(D_{(i}X_{j)}-D^lX_l h_{ij}+F_{ij})
\end{array}\right).
\]
Further, we consider the composition
$\mathcal{P}\circ\mathcal{P}^*:\mathscr{C}\times\mathscr{X}\rightarrow \mathscr{C}\times \mathscr{X}$, given by
\begin{equation*}
\hspace{-5mm}\mathcal{P}\circ\mathcal{P}^*\left(\begin{array}{c}
X\\
X_i
\end{array}\right) \equiv \left(\begin{array}{c}
2\Delta_\bmh\Delta_\bmh X-r^{ij}D_iD_jX+2r\Delta_\bmh X+\tfrac{3}{2}D^i rD_iX+(\tfrac{1}{2}\Delta_\bmh r+r_{ij}r^{ij})X\\
+D^iD^j H_{ij}-\Delta_\bmh H_k{}^k-r^{ij}H_{ij}+\bar{H} \\[1em]
D^j\Delta_\bmh D_{(i}X_{j)}+D_i\Delta_\bmh D^kX_k+D^j\Delta_\bmh F_{ij}-D_i\Delta_\bmh F_k{}^k-\bar{F}_i
\end{array}\right)
\end{equation*}
where
\begin{align*}
& \bar{H}\equiv 2(K\bar{Q}-K^{ij}\bar{Q}_{ij})+2(K^{ki}K^j{}_k-KK^{ij})\bar{\gamma}_{ij},\\
& \bar{F}_i\equiv \left(D_i K^{kj}-D^k K^j{}_i\right)\bar{\gamma}_{jk}-\left(K^k{}_i D^j-\tfrac{1}{2}K^{kj}D_i\right)\bar{\gamma}_{jk}+\tfrac{1}{2}K^k{}_i D_k \bar{\gamma}\\
&\bar{\gamma}_{ij}\equiv D_iD_jX-X r_{ij}-\Delta_\bmh X h_{ij}+H_{ij}\\
&\bar{Q}_{ij}\equiv -\Delta_\bmh(D_{(i}X_{j)}-D^kX_k h_{ij}+F_{ij})
\end{align*}
and $F_{ij},~H_{ij}$ as above. 
\medskip
One has the following:

\begin{lemma}
The operator
$\mathcal{P}\circ\mathcal{P}^*:\mathscr{C}\times\mathscr{X}\rightarrow
\mathscr{C}\times \mathscr{X}$ as defined above is a self-adjoint fourth
order elliptic operator.
\end{lemma}

\begin{proof}
The self-adjointness follows from the definition as the operator is
obtained by the composition of an operator and its formal adjoint. To verify
the ellipticity of the operator we notice that the symbol is given by 
\[
\sigma_\xi\left(\begin{array}{c}
X\\
X_i
\end{array}\right)=\left(\begin{array}{c}
2\vert\xi\vert^2X\\
\xi^j\vert\xi\vert^2\xi_{(i}X_{j)}+\xi_i\vert\xi\vert^2 \xi_jX^j
\end{array}
\right)
\]
for $\xi_i$ a covector and $|\xi|^2 \equiv \delta_{ij}\xi^i
\xi^j$. Clearly,  the first component is an isomorphism if
$|\xi|^2\neq 0$. For the second component, contract first with $\xi^i$
to get $2\vert\xi\vert^4 \xi^j X_j=0$ for $X_i$ in the kernel, which
implies $\xi^jX_j=0$. Substituting back into the symbol, one obtains
that $\vert\xi\vert^4 X_i=0$. So, for $|\xi|^2\neq 0$, the symbol is
injective. Clearly the codomain has the same dimension as the domain,
and therefore $\sigma_\xi$ is an isomorphism for $|\xi|^2\neq 0$
---i.e. $\mathcal{P}\circ\mathcal{P}^*$ is fourth-order elliptic operator.
\end{proof}

\medskip
The previous discussion suggests the following:

\begin{definition}
The equation
\begin{equation}
\mathcal{P}\circ\mathcal{P}^*\left(\begin{array}{c}
X\\
X_i
\end{array}\right)=0
\label{ApproximateKIDEquation}
\end{equation}
will be called the approximate Killing initial data (KID) equation and a
solution $(X,X^i)$ thereof an approximate Killing initial data set ---or
approximate KID for brevity.
\end{definition}

\begin{remark}
{\em
As pointed out in \cite{Dai04c}, the equation $\mathcal{P}\circ\mathcal{P}^*(X,X_i)=0$ is the Euler--Lagrange equation of the action
\begin{equation*}
\int_{\mathcal{U}}\mathcal{P}^*\left(\begin{array}{c}
X\\
X_i
\end{array}\right)\cdot \mathcal{P}^*\left(\begin{array}{c}
X\\
X_i
\end{array}\right)~\mbox{d}\mu
\end{equation*}
Note that, had we used the operator $D\Phi^*$ rather than $\mathcal{P}^*$, then the pointwise norm defined by the integrand would contain terms of inconsistent physical dimension: $\left[ X\right]=L^{-2}, ~\left[X_i\right]=1$, and so for instance $\left[(D_i D_j X)(D^i D^j X)\right]=L^{-6}$, while $\left[D_{(i}X_{j)}D^{(i} X^{j)}\right]=L^{-4}$.}
\end{remark}

\subsection{Integration by parts identities}
\label{Subsection:IntegrationParts}
The expressions in the previous subsection and several of our
arguments in latter parts are based on integration by parts. For quick
reference, in this subsection we provide the integral expressions
relating the operators $\mathcal{P}$ and $\mathcal{P}^*$ including
boundary terms.

\medskip
Let $\mathcal{U}\subset \mathcal{S}$ denote a compact set with
boundary $\partial \mathcal{U}$. Recall that by definition 
\begin{eqnarray*}
 \int_{\mathcal{U}} 
\left( 
\begin{array}{c}
X \\
X^i
\end{array}
 \right)\cdot
\mathcal{P}
\left( 
\begin{array}{c}
\gamma_{ij} \\
q_{kij}
\end{array}
\right) \mbox{d}\mu &=& \int_{\mathcal{U}}  
\left( 
\begin{array}{c}
X \\
X^i
\end{array}
 \right)\cdot
D\Phi
\left( 
\begin{array}{c}
\gamma_{ij} \\
-D^k q_{kij}
\end{array}
\right) \mbox{d}\mu \\
& = & \int_{\mathcal{U}}  
\left( 
\begin{array}{c}
X \\
X^i
\end{array}
 \right)\cdot
\left( 
\begin{array}{c}
D^iD^j \gamma_{ij} - r_{ij} \gamma^{ij} -\Delta_\bmh \gamma + H  \\
D^j D^k q_{kij} -D_i D^k q_{kj}{}^j - F_i  
\end{array}
\right) \mbox{d}\mu\\
& = & \int_{\mathcal{U}} X \left(  D^iD^j \gamma_{ij} - r_{ij}
      \gamma^{ij} -\Delta_\bmh \gamma + H\right) \mbox{d}\mu \\
&& \hspace{2cm} +
      \int_{\mathcal{U}} X^i \left( D^j D^k q_{kij} -D_i D^k
   q_{kj}{}^j - F_i \right) \mbox{d}\mu\\
& = & J_1 + J_2.
\end{eqnarray*}
We now proceed to use integration by parts on $J_1$ and $J_2$. A
lengthy computation shows that 
\begin{eqnarray*}
J_1 &\equiv &  \int_{\mathcal{U}} X \left(  D^iD^j \gamma_{ij} - r_{ij}
      \gamma^{ij} -\Delta_\bmh \gamma + H\right) \mbox{d}\mu\\
& = & \int_{\mathcal{U}} \gamma_{ij} \left( D^j D^iX -h^{ij} \Delta_\bmh
      X - X r^{ij} + 2(K^{ki}K^j{}_k-KK^{ij}
)\right)\mbox{d}\mu\\
      &&\hspace{1cm}+ \int_{\mathcal{U}} 2q_{kij}(h^{ij}XD^k K +h^{ij}K D^kX-XD^kK^{ij}-K^{ij}D^kX)d\mu \\
      &&\hspace{1cm}+ \oint_{\partial \mathcal{U}} n^k (\mathcal{A}_k+\mathcal{B}_k) \mbox{d}S
\end{eqnarray*}
where the boundary integrands are given by
\begin{eqnarray*}
&& \mathcal{A}_k \equiv X D^j \gamma_{jk} - D^j X \gamma_{jk} - D_k X
\gamma -X D_k \gamma,\\
&& \mathcal{B}_k \equiv 2(K^{ij}q_{kij}-K q_{kj}{}^j)X.
\end{eqnarray*}
Similarly, one finds that
\begin{eqnarray*}
J_2 & \equiv &  \int_{\mathcal{U}} X^i \left( D^j D^k q_{kij} -D_i D^k
   q_{kj}{}^j - F_i \right) \mbox{d}\mu \\
&=& \int_{\mathcal{U}} q_{kij} \left( D^k D^j X^i + D^k D_l
    h^{ij}\right)\mbox{d}\mu \\
&& \hspace{1cm}-\int_{\mathcal{U}} \gamma_{jk}\left(
    (D_i K^{kj} - D^k K^j{}_i)X^i + D^j(X^iK_i{}^k) -
    \tfrac{1}{2}D_i(X^i K^{kj})- \tfrac{1}{2}h^{jk} D_i (X^l K_l{}^i)
    \right)\mbox{d}\mu \\
&& \hspace{1cm}+ \oint_{\partial\mathcal{U}} n^k (\mathcal{C}_k+\mathcal{D}_k) \mbox{d}S 
\end{eqnarray*}
where the boundary integrands are given by
\begin{eqnarray*}
&& \mathcal{C}_k \equiv X^i D^l q_{lik} - D^j X^i q_{kij} + D_i X^i q_{kj}{}^j -
   X_i D^l q_{lj}{}^j, \\
&& \mathcal{D}_k \equiv X^i K_i{}^l \gamma_{kl} -\tfrac{1}{2}X_k K^{lj}
   \gamma_{jl} - \tfrac{1}{2}X^iK_{ik}\gamma. 
\end{eqnarray*}
Putting everything together and after some further manipulations one
finds the identity
\begin{equation}
\int_{\mathcal{U}} 
\left( 
\begin{array}{c}
X \\
X^i
\end{array}
 \right)\cdot
\mathcal{P}
\left( 
\begin{array}{c}
\gamma_{ij} \\
q_{kij}
\end{array}
\right) \mbox{d}\mu =
\int_{\mathcal{U}} \left(
\begin{array}{c}
\gamma^{ij} \\
q^{kij} 
\end{array}\right) \cdot
\mathcal{P}^*
\left(
\begin{array}{c}
X \\
X_i
\end{array}
\right) + \oint_{\partial\mathcal{U}} n^k \big( \mathcal{A}_k + \mathcal{B}_k + \mathcal{C}_k + \mathcal{D}_k \big)
\mbox{d}S.
\label{IntegrationByParts}
\end{equation}

\section{The approximate Killing vector equation on asymptotically
  Euclidean manifolds}
\label{Section:AsymptoticallyEuclidean}
In this section we study the solvability of the approximate KID
equation on asymptotically Euclidean manifolds. The standard methods
to study elliptic equations on this type of manifolds employ so-called
\emph{weighted Sobolev spaces} ---thus, we start by briefly reviewing our
basic technical tools in Section \ref{Section:SobolevSpaces}. The key
assumption on the class of initial data sets to be considered are
discussed in \ref{Section:DecayAssumptions}. The
existence results for the approximate KID equation are given in
Subsection \ref{ApproximateKIDEquation}.

\subsection{Weighted Sobolev spaces}
\label{Section:SobolevSpaces}
In order to discuss the decay of the various tensor fields in the
3-manifold $\mathcal{S}$ we need to make use of \emph{weighted Sobolev
spaces} ---see e.g. \cite{Can81,Loc81,ChoChr81a,Bar86}. Given an
arbitrary point $p\in \mathcal{S}$ one defines for $x\in \mathcal{S}$
\[
\sigma(x) \equiv \big(1+ d(p,x)^2\big)^{1/2}
\]
where $d(p,x)$ denotes the Riemannian distance on $\mathcal{S}$. The
function $\sigma$ is used to define the weighted $L^2$-norm
\[
\parallel u \parallel_\delta \equiv \left( \int_{\mathcal{S}} |u|^2
  \sigma^{-2\delta-3} \mbox{d}^3 x \right)^{1/2}, \qquad \delta \in \mathbb{R}.
\]
In particular, if $\delta =-3/2$ one recovers the usual
$L^2$-norm. Different choices of origin give rise to equivalent
weighted norms. 

\begin{remark}
{\em In the above and in the rest of the article, we follow Bartnik's
  conventions \cite{Bar86} to denote the weighted Sobolev spaces and
  norms. Different choices of the point $p$ give rise to equivalent
  weighted norms ---see e.g. \cite{Bar86,ChoChr81a}. Thus, in a slight
  abuse of notation, we denote all these equivalent norms with the same symbol.}
\end{remark}

The fall-off behaviour of the various fields will be expressed in
terms of weighted Sobolev spaces $H^s_\delta$ consisting of functions
for which 
\[
\parallel u \parallel_{s,\delta} \equiv \sum_{0\leq |\alpha| \leq
  s} \parallel D^\alpha u \parallel_{\delta -|\alpha|} < \infty,
\]
where $s$ is a non-negative integer, and where
$\alpha=(\alpha_1,\alpha_2,\alpha_3)$ is a multiindex,
$|\alpha|=\alpha_1+\alpha_2+\alpha_3$. One says that $u\in H^\infty_s$
if $u\in H^s_\delta$ for all $s$. We will say that a tensor belongs to a given function space if its norm does. 

\medskip
In the following given some coordinates $x=(x^\alpha)$, let $|x|^2\equiv
\delta_{\alpha\beta}x^\alpha x^\beta$. We will make repeated use of the following result\footnote{Recall that
$f(x)=o(|x|^\alpha)$ if $f(x)/|x|^{\alpha}\rightarrow 0$ as $|x|\rightarrow0$. If
$\partial^n f(x) =o(|x|^{\alpha-n})$ for each non-negative integer,
then we write $f(x)=o_\infty(|x|^\alpha)$.}:

\begin{lemma}
Let $u\in H^\infty_\delta$. Then $u$ is smooth (i.e. $C^\infty$) over
$\mathcal{S}$ and has a fall-off at infinity such that
\[
D^l u = o\big(|x|^{\delta -|l|}\big).
\]
\end{lemma}

The proof can be found in \cite{Bar86} ---see also Section 6.1 in
\cite{BaeVal10b}. The following \emph{Multiplication Lemma} has been proven in
\cite{BaeVal10b}:

\begin{lemma}
Let $u=o_\infty(|x|^{\delta_1})$, $v=o_\infty(|x|^{\delta_2})$ and
$w=O(|x|^\gamma)$. Then 
\[
uv = o_\infty (|x|^{\delta_1+\delta_2}), \qquad uw = o_\infty(|x|^{\delta_1+\gamma}).
\]
\end{lemma}

\begin{remark}
{\em This lemma can be readily extended to tensor fields.}
\end{remark}

\subsection{Decay assumptions}
\label{Section:DecayAssumptions}
In what follows we will consider initial data sets
$(\mathcal{S},h_{ij},K_{ij})$ for the vacuum Einstein field equations
possessing, in principle, several \emph{asymptotically Euclidean
ends}. Thus, we assume there exists a compact set $\mathcal{B}$ such that
\[
\mathcal{S}\setminus \mathcal{B}= \sum_{k=1}^n \mathcal{S}_{(k)} 
\]
where $\mathcal{S}_{(k)}$, $k=1,\ldots,n$, are open sets diffeomorphic to the complement of
a closed ball on $\mathbb{R}^3$. Each set $\mathcal{S}_{(k)}$ is
called an \emph{asymptotic end}.  On
each of these ends one can introduce (non-unique) \emph{asymptotically Cartesian}
coordinates $x=(x^\alpha)$. Our basic decay assumptions for the fields
$h_{ij}$ and $K_{ij}$ are expressed in terms of these coordinates:

\begin{assumption}[Decay Assumptions]
\label{Decay:Assumptions}
{\em On each asymptotically Euclidean end one has
\begin{eqnarray*}
&& h_{\alpha\beta} - \delta_{\alpha\beta} = o_\infty(|x|^{-1/2}), \\
&& K_{\alpha\beta}= o_\infty(|x|^{-3/2}).
\end{eqnarray*}
We assume further that the decay is such that the ADM 4-momentum is non-vanishing.}
 \end{assumption}

The following definition will prove useful:

\begin{definition}
\label{Definition:StationaryData}
An asymptotically Euclidean initial data set
$(\mathcal{S},h_{ij},K_{ij})$ satisfying the  Decay
Assumptions \ref{Decay:Assumptions} is said to be \emph{stationary} if
there exists non-trivial $(N,N^i)\in H^2_{1/2}$ such that
\begin{equation}
\mathcal{P}^*
\left(
\begin{array}{c}
N \\
N^i
\end{array}
\right) =0,
\label{StationaryEquation}
\end{equation}
\end{definition}

\begin{remark}
\label{StationaryDataSetAdmitsKID}
{\em As it is to be expected, a stationary initial data set (in the sense
of Definition \ref{Definition:StationaryData}) admits a KID. To see
this, observe that equation \eqref{StationaryEquation} implies that
\begin{subequations}
\begin{eqnarray}
&& D_i D_j N - N r_{ij} -\Delta_\bmh N h_{ij} + H_{ij}=0, \label{StationaryEq1}\\
&& D_k\big( D_{(i} N_{j)} - D^l N_l h_{ij} +F_{ij} \big)=0. \label{StationaryEq2}
\end{eqnarray}
\end{subequations}
Direct inspection shows that, if one assumes $(N,N^i)\in H^2_{1/2}$ in addition to decay Assumption \ref{Decay:Assumptions}, then
\[
D_{(i} N_{j)} - D^l N_l h_{ij} +F_{ij} =o(|x|^{-1/2})
\]
---i.e. the tensor field $D_{(i} N_{j)} - D^l N_l h_{ij} +F_{ij}$ vanishes at infinity. Since it is also covariantly constant as a consequence of (\ref{StationaryEq2}), it follows then that it must vanish identically  ---i.e.
\[
D_{(i} N_{j)} - D^l N_l h_{ij} +F_{ij}=0.
\]
Combining this observation with (\ref{StationaryEq1}), we see that
$D\Phi^*(N,N^i)=0$ and hence $(N,-\tfrac{1}{2}N^i)$ solves the KID
equations (\ref{KID1})-(\ref{KID2}) ---see Remark
\ref{Rmk:EquivalenceOfKIDsandDPhiStar}.  Finally, we observe that the
behaviour $(N,-\tfrac{1}{2}N^i)=o(|x|^{1/2})$ for a KID is only
consistent with {\em translational Killing vector fields} ---i.e. Killing
vectors which to leading order look like a (timelike,
spatial or null) translation in the Minkowski spacetime. Now, the only
type of translational Killing vector a spacetime with non-vanishing
ADM 4-momentum can admit is one which is timelike and bounded at
infinity ---i.e. a stationary Killing vector, see Section III in
\cite{BeiChr96}. Clearly the reverse is also true: if an initial data
set admits a stationary Killing vector, then the data is stationary in
the sense of Definition 
\ref{Definition:StationaryData}. It should be stressed that \emph{our
definition of stationary initial data sets excludes initial data sets
for the Minkowski spacetime} as these necessarily have a vanishing ADM
4-momentum. The condition on the ADM 4-momentum in Definition
\ref{Definition:StationaryData} arises from the need to single out the
stationary Killing vector field from among the collection of
translational Killing vectors.}
\end{remark}


The asymptotic behaviour of solutions to the KID equations has been
studied in \cite{BeiChr96} from where we adapt the following result:

\begin{proposition}
\label{Proposition:AsymptoticBehaviourKIDs}
Let $(\mathcal{S},h_{ij},K_{ij})$ denote a smooth vacuum initial data set
satisfying the Decay Assumptions \ref{Decay:Assumptions}. Moreover,
let $N$, $Y^i$ be, respectively, a smooth scalar field and a vector
field over $\mathcal{S}$ satisfying the KID equations. Then, there
exists a constant tensor with components $\mathfrak{L}_{\mu\nu}=\mathfrak{L}_{[\mu\nu]}$ such
that
\[
N -\mathfrak{L}_{0\alpha} x^\alpha = o_\infty( |x|^{1/2}), \qquad
Y^\alpha - \mathfrak{L}_{\alpha\beta}x^\beta = o_\infty( |x|^{1/2}).
\]
If $\mathfrak{L}_{\mu \nu}=0$, then there exists a constant vector with
components $\mathfrak{A}^\mu$ such that
\[
N - \mathfrak{A}^0 = o_\infty (|x|^{-1/2}), \qquad Y^\alpha - \mathfrak{A}^\alpha = o_\infty(|x|^{-1/2}).
\]
Finally, if $\mathfrak{L}_{\mu\nu}=\mathfrak{A}^\mu=0$, then $N=0$ and $Y^i=0$. 
\end{proposition}

\subsection{Basic results of the theory of elliptic equations on
  asymptotically Euclidean manifolds}
\label{Section:Ellipticity}
In view of the Decay Assumptions \ref{Decay:Assumptions}, the approximate KID equation \eqref{ApproximateKIDEquation} can be
written, in local coordinates, in the form
\[
\mathcal{L}\mathbf{u}\equiv (\mathbf{A}^{\alpha\beta\gamma\delta} +
\mathbf{a}^{\alpha\beta\gamma\delta})\cdot\partial_\alpha\partial_\beta\partial_\gamma\partial_\delta
\mathbf{u} +
\mathbf{a}^{\alpha\beta\gamma}\cdot\partial_\alpha\partial_\beta\partial_\gamma
\mathbf{u} +\mathbf{a}^{\alpha\beta}\cdot\partial_\alpha\partial_\beta \mathbf{u}
+\mathbf{a}^\alpha\cdot\partial_\alpha \mathbf{u} + \mathbf{a} \cdot\mathbf{u}=0,
\]
where $\mathbf{u}:\mathcal{S}\rightarrow \mathbb{R}^4$ is a
vector-valued function over $\mathcal{S}$,
$\mathbf{A}^{\alpha\beta\gamma\delta}$ denote constant matrices, while $\mathbf{a}^{\alpha\beta\gamma\delta}$, $\mathbf{a}^{\alpha\beta\gamma}$, $\mathbf{a}^{\alpha\beta}$,
$\mathbf{a}^{\alpha}$ and $\mathbf{a}$ denote smooth matrix-valued functions of the coordinates
$x=(x^\alpha)$. 

\medskip
The operator $\mathcal{L}$ is said to be \emph{asymptotically
  homogeneous} if 
\[
\mathbf{a}^{\alpha\beta\gamma\delta}\in H^\infty_\tau, \qquad
\mathbf{a}^{\alpha\beta\gamma}\in H^\infty_{\tau-1}, \qquad
\mathbf{a}^{\alpha\beta}\in H^\infty_{\tau-2}, \qquad
\mathbf{a}^\alpha\in H^\infty_{\tau-3}, \qquad \mathbf{a} \in H^\infty_{\tau-4},
\]
for some $\tau<0$ ---see e.g. \cite{Can81,Loc81}. 

\begin{remark}
{\em Direct inspection using the Decay Assumptions
\ref{Decay:Assumptions} imply that $\mathcal{L}$ is asymptotically
homogeneous with $\tau=-1/2$. This is the standard assumption when working with
weighted Sobolev spaces.} 
\end{remark}

\medskip
In the following we will make use of the following version of the
\emph{Fredholm alternative} for fourth-order asymptotically homogeneous
operators on asymptotically Euclidean manifolds ---see \cite{Can81}:

\begin{proposition}
Let $\mathcal{L}$ be an asymptotically homogeneous elliptic operator
of order 4 with smooth coefficients. Given $\delta$ not a negative
integer, the equation
\[
\mathcal{L} \mathbf{u} = \mathbf{f}, \qquad \mathbf{f}\in H^0_{\delta-4}
\]
has a solution $\mathbf{u}\in H^4_\delta$ if and only if 
\[
\int_{\mathcal{S}} \mathbf{f}\cdot \mathbf{v}\,\mbox{\em d}\mu=0
\]
for all $\mathbf{v}$ satisfying 
\[
\mathcal{L}^* \mathbf{v} =0, \qquad \mathbf{v}\in H^0_{1-\delta},
\]
where $\mathcal{L}^*$ denotes the formal adjoint of $\mathcal{L}$. 
\end{proposition}

Finally, to assert the regularity of solutions we need the following
elliptic estimate ---see Theorem 6.3. of \cite{Can81}:

\begin{proposition}
\label{Prop:EllipticRegularity}
Let $\mathcal{L}$ be an asymptotically homogeneous elliptic operator
of order 4 with smooth coefficients. Then for any $\delta\in
\mathbb{R}$ and any $s\geq 4$, there exists a constant $C$ such that
for every $\mathbf{v}\in H^s_{loc}\cap H^0_\delta$, the following
inequality holds:
\[
\parallel \mathbf{v}\parallel_{H^s_\delta} \leq
C\big( \parallel\mathcal{L}\mathbf{v} \parallel_{H^{s-4}_{\delta-2}} +
\parallel \mathbf{v} \parallel_{H^{s-4}_{\delta}}\big).
\]
\end{proposition}

In the above proposition $H^s_{loc}$ denotes the local Sobolev space
 ---that is, $\mathbf{v}\in H^s_{loc}$ if for an arbitrary smooth
 function $\phi$ with compact support, $\phi\mathbf{v}\in H^s$.

\begin{remark}
{\em If  $\mathcal{L}$ has smooth coefficients and
  $\mathcal{L}\mathbf{v} =0$, then it follows that all the
  $H^s_\delta$ norms of $\mathbf{v}$ are bounded by the $H^0_\delta$
  norm. Thus, it follows that if a solution to
  $\mathcal{L}\mathbf{v}=0$ exists, it must be smooth
  ---\emph{elliptic regularity}. }
\end{remark}

\subsection{Existence of solutions to the approximate Killing vector equation}
\label{Section:ExistenceSolutions}

We are now in the position of analysing the existence of solutions to
the approximate Killing equation \eqref{ApproximateKIDEquation}. Our
main tools will be the Fredholm alternative and integration by
parts. We begin by considering some auxiliary results.

\subsubsection{Auxiliary existence results}
The following result relating solutions to the approximate Killing
equations to solutions to the KID equations will be needed in our main
result:
\begin{lemma}
\label{Lemma:StationarySolutions}
Let $(\mathcal{S},h_{ij},K_{ij})$ be a complete, smooth
asymptotically Euclidean initial data set for the Einstein vacuum
field equations with $n$ asymptotic ends and satisfying the Decay
Assumptions \ref{Decay:Assumptions}. Then, for $0<\beta\leq 1/2$, 
\[\text{ker}\lbrace \mathcal{P}\circ\mathcal{P}^*:H^\infty_\beta\rightarrow H^\infty_{\beta-4}\rbrace=\text{ker}\lbrace\mathcal{P}^*:H^\infty_\beta\rightarrow H^\infty_{\beta-2}\rbrace \]
That is to say, the equation 
\[
\mathcal{P}\circ \mathcal{P}^* 
\left(
\begin{array}{c}
N \\
N^i
\end{array}
\right) =0
\]
admits a solution $(N,N^i)\in
H^\infty_\beta$, $0<\beta\leq 1/2$, if and only if
$(\mathcal{S},h_{ij},K_{ij})$ is stationary in the sense of Definition \ref{Definition:StationaryData}. Moreover, if the solution exists then it is unique up to constant rescaling.
\end{lemma}
\begin{proof}
Assume that $\mathcal{P}\circ \mathcal{P}^*(N,N^i)=0$. Making use of
the identity \eqref{IntegrationByParts} with 
\[
\left(
\begin{array}{c}
\gamma^{ij}\\
q^{kij}
\end{array}
\right) = 
\mathcal{P}^* 
\left(
\begin{array}{c}
N \\
N^i
\end{array}
\right)
\]
one finds that 
\[
\int_{\mathcal{S}} 
\mathcal{P}^* 
\left(
\begin{array}{c}
N \\
N^i
\end{array}
\right)\cdot
\mathcal{P}^* 
\left(
\begin{array}{c}
N \\
N^i
\end{array}
\right) \mbox{d}\mu = -\oint_{\partial\mathcal{S}_\infty}n^k (\mathcal{A}_k +
\mathcal{B}_k +\mathcal{C}_k +\mathcal{D}_k) \mbox{d}S
\]
where $\partial \mathcal{S}_\infty$ denotes the sphere at infinity. We
proceed now to evaluate the various boundary terms.

\medskip
We observe that under the Decay
Assumptions \ref{Decay:Assumptions} direct inspection shows
that
\[
H_{ij} = o(|x|^{-2}),
\]
from where it follows that
\begin{eqnarray*}
&& \gamma_{ij} = D_i D_j N - N r_{ij} -\Delta N h_{ij} + H_{ij}= o(|x|^{-3/2}).
\end{eqnarray*}
Hence, one has that
\begin{eqnarray*}
&& \mathcal{A}_k = N D^i\gamma_{ik} - D^iN \gamma_{ik} + D_k N\gamma -N
   D_k\gamma= o(|x|^{-2}).
\end{eqnarray*}
Thus, taking into account that $\mbox{d} S = O(|x|^2)$ one concludes
that
\[
\oint_{\partial\mathcal{S}_\infty} n^k \mathcal{A}_k \mbox{d}S=0.
\]

Next, we consider
\begin{eqnarray*}
&& \mathcal{C}_k =N^i D^l q_{lik} - D^j N^i q_{kij} + D_i N^i q_{kj}{}^j -
   N_i D^l q_{lj}{}^j
\end{eqnarray*}
where
\[
q_{kij} = D_k\big( D_{(i} N_{j)} - D^l N_l h_{ij} - F_{ij}\big),
\qquad F_{ij} =2N (K h_{ij} - K_{ij}).
\]
From the Decay
Assumptions \ref{Decay:Assumptions} it follows that in this case
\[
F_{ij} =o (|x|^{-1}), \qquad q_{kij} = o(|x|^{-3/2})
\]
so that
\[
\mathcal{C}_k = o(|x|^{-2}).
\]
Thus, one has that
\[
\oint_{\partial \mathcal{S}_\infty} \mathcal{C}_k n^k \mbox{d}S =0.
\]

\medskip
Finally, similar considerations give that
\begin{eqnarray*}
&& \mathcal{D}_k = \frac{1}{2}N_k K^{lj} \gamma_{jl} + \frac{1}{2}N^i K_{ik}
   \gamma - N^i K_i{}^l \gamma_{kl} = o(|x|^{-5/2})
\end{eqnarray*}
so that
\[
\oint_{\partial\mathcal{S}_\infty} n^k \mathcal{D}_k \mbox{d}S=0.
\]

\medskip
From the previous discussion it follows then that
\[
\int_{\mathcal{S}} 
\mathcal{P}^* 
\left(
\begin{array}{c}
N \\
N^i
\end{array}
\right)\cdot
\mathcal{P}^* 
\left(
\begin{array}{c}
N \\
N^i
\end{array}
\right) \mbox{d}\mu =0
\]
so that $\mathcal{P}^*(N,N^i)=0$, and therefore that the data is
stationary. Finally, uniqueness of the solution follows from
Proposition \ref{Proposition:AsymptoticBehaviourKIDs}. Suppose, for
contradiction, that there exist two distinct solutions, giving rise to
two distinct KID sets $(N,-\tfrac{1}{2}N^i)$ and
$(\tilde{N},-\tfrac{1}{2}\tilde{N}^i)$. Taking the appropriate linear
combination we arrive at a KID set with a lapse that goes to zero at
infinity while the shift is in $H^\infty_\beta$, $\beta\leq 1/2$
---that is, one has a KID associated to a spatial translation. This
contradicts the fact that the ADM 4-momentum of the initial data is
non-vanishing ---see Section III in \cite{BeiChr96}.
\end{proof}


\begin{remark}
{\em Making use of the asymptotic expansion provided by Proposition \ref{Proposition:AsymptoticBehaviourKIDs}
one finds that for stationary initial data sets, the solutions
provided by Lemma \ref{Lemma:StationarySolutions} are of the form: 
\begin{equation}
 N-\mathfrak{A}^0 = o_\infty(|x|^{-1/2}), \qquad N^\alpha-\mathfrak{A}^\alpha =o_\infty(|x|^{-1/2})
\label{KillingVectorAsymptotics}
\end{equation}
with the components of a $\mathfrak{A}^\mu$ a constant vector field.
}
\end{remark}

\subsubsection{Main existence result}
Following \cite{Dai04c} we now will look for solutions of the approximate
Killing equation  such that  
 \begin{eqnarray*}
 && X = \lambda |x| + \vartheta, \qquad \vartheta \in    H^\infty_{1/2},\\
&& X^i \in H^\infty_{1/2}.
 \end{eqnarray*}
in each asymptotically Euclidean end and where $\lambda$ is a
constant. This Ansatz is motivated by the observation that
$\Delta^2_\bmdelta |x|=0$, with $\Delta_\bmdelta$ the flat Laplacian ---that is, the blowing up term $\lambda
|x|$ is in the kernel of the first component of the operator
$\mathcal{P}\circ\mathcal{P}^*$ evaluated on the 3-dimensional flat metric. 
\begin{theorem}
\label{MainResult}
Let $(\mathcal{S},h_{ij},K_{ij})$ be a complete, smooth
asymptotically Euclidean initial data set for the Einstein vacuum
field equations with $n$ asymptotic ends,  satisfying the Decay
Assumptions \ref{Decay:Assumptions}. Then there exists a 
solution $(X,X^i)$ to the approximate KID equation,
\[
\mathcal{P}\circ \mathcal{P}^* 
\left(
\begin{array}{c}
X \\
X^i
\end{array}
\right) =0,
\]
such that at each asymptotic end one has the asymptotic
behaviour
\begin{eqnarray*}
 && X_{(k)} = \lambda_{(k)} |x| + \vartheta_{(k)}, \qquad \vartheta_{(k)} \in    H^\infty_{1/2},\\
&& X^i_{(k)} \in H^\infty_{1/2},
 \end{eqnarray*}
where $\lambda_{(k)}$, $k=1,\ldots,n$, are constants and $\lambda_{(k)}=0$ for some $k$
if and only if $(\mathcal{S},h_{ij},K_{ij})$ is stationary in the
sense of Definition \ref{Definition:StationaryData}. Moreover, the solution is unique up to constant rescaling.
\end{theorem}
\begin{proof}
Substituting the above Ansatz in
equation \eqref{ApproximateKIDEquation} one obtains
\begin{equation}
\mathcal{P}\circ \mathcal{P}^* 
\left(
\begin{array}{c}
\vartheta \\
X^i
\end{array}
\right) =
-\mathcal{P}\circ \mathcal{P}^* 
\left(
\begin{array}{c}
\lambda |x| \\
0
\end{array}
\right).
\label{InhomogeneousEquation}
\end{equation}
Under the Decay Assumptions \ref{Decay:Assumptions}, a lengthy
computation shows that
\begin{eqnarray*}
&H_{ij}= o(|x|^{-2}), \qquad F_{ij} = o(|x|^{-1/2}), \qquad Q_{ij} =
  o(|x|^{-5/2}), \qquad \bar{\gamma}_{ij} = o(|x|^{-1}),&\\
& F_i = o(|x|^{-7/2}), \qquad H =o(|x|^{-4}),&
\end{eqnarray*}
where, in particular, it has been used that
\[
\partial_\alpha |x| = \frac{x_\alpha}{|x|} = O(1),
\qquad \partial_\alpha\partial_\beta |x| =
\frac{\delta_{\alpha\beta}}{|x|}- \frac{x_\alpha x_\beta }{|x|^3}
=O(|x|^{-1}).
\]
Hence, 
\begin{eqnarray*}
&& 2\Delta_\bmh\Delta_\bmh X-r^{ij}D_iD_jX+2r\Delta_\bmh X+D^iD^j
   H_{ij}-\Delta_\bmh H_k{}^k+\bar{H}\\
&& \hspace{3cm} +\tfrac{3}{2}D^i
   rD_iX+(\tfrac{1}{2}\Delta_\bmh r+r_{ij}r^{ij})X-r^{ij}H_{ij} =o(|x|^{-7/2}),\\
&& D^j\Delta_\bmh F_{ij}-D_i\Delta_\bmh F_k{}^k-\bar{F}_i = o(|x|^{-7/2}).
\end{eqnarray*}
so that
\[
\mathcal{P}\circ \mathcal{P}^* 
\left(
\begin{array}{c}
\lambda |x| \\
0
\end{array}
\right) \in H^0_{-7/2}.
\]
To prove the existence of solutions to equation
\eqref{InhomogeneousEquation} we make use of the Fredholm
alternative in weighted Sobolev spaces, according to which equation
\eqref{InhomogeneousEquation} will have solution $(\vartheta,X^i)$ if and only
if its right-hand-side is $L^2$-orthogonal to $\text{coker}\lbrace\mathcal{P}\circ\mathcal{P}^*:H^4_{1/2}\rightarrow H^0_{1/2}\rbrace $ ---i.e. if and only if
\[
\int_{\mathcal{S}} \mathcal{P}\circ \mathcal{P}^* 
\left(
\begin{array}{c}
\lambda |x| \\
0
\end{array}
\right) \cdot
\left(
\begin{array}{c}
N \\
N^i
\end{array}
\right) \mbox{d}\mu =0
\]
for all $(N,N^i)\in H^0_{1/2}$ for which
\[
\mathcal{P}\circ \mathcal{P}^* 
\left(
\begin{array}{c}
N \\
N^i
\end{array}
\right) =0.
\]
From Lemma \ref{Lemma:StationarySolutions} we know that this equation
has non-trivial solutions (i.e. that $\text{coker}\lbrace\mathcal{P}\circ\mathcal{P}^*:H^4_{1/2}\rightarrow H^0_{1/2}\rbrace $ will be non-trivial) if and only if $(\mathcal{S},h_{ij},K_{ij})$
is stationary. Thus, if the initial data set is not stationary, the Fredholm alternative guarantees a solution $(\vartheta, X^i)$ to \eqref{InhomogeneousEquation}. 
\\

For the stationary case, the cokernel is spanned by a single Killing vector with components $(N,N_i)$, taking the form of (\ref{KillingVectorAsymptotics}). Let
\[\left(\begin{array}{c}
\Gamma_{ij}\\
Q_{kij}
\end{array}\right)\equiv \mathcal{P}^*\left(\begin{array}{c}
\lambda\vert x\vert\\
0
\end{array}\right)=\left(\begin{array}{c}
 \lambda(D_i D_j\vert x\vert -\vert x\vert r_{ij}-\Delta\vert x\vert h_{ij})+H_{ij}\\
D_k F_{ij}
\end{array}\right)\]
where, now, 
\begin{eqnarray*}
&& H_{ij}\equiv
  2\lambda\vert x\vert(K^k{}_iK_{jk}-KK_{ij})=o(\vert x\vert^{-2}),\\
&& F_{ij}\equiv 2\lambda\vert x\vert(Kh_{ij}-K_{ij})=o(\vert x\vert^{-1/2}).
\end{eqnarray*}
It then follows that
\[\Gamma_{ij}=o(\vert x\vert^{-1}),\qquad Q_{kij}=o(\vert x\vert^{-3/2})\]
and that 
\[\mathcal{P}\circ\mathcal{P}^*\left(\begin{array}{c}
\lambda\vert x\vert\\
0
\end{array}\right)=\mathcal{P}\left(\begin{array}{c}
\Gamma_{ij}\\
Q_{kij}
\end{array}\right)=o(\vert x\vert^{-7/2})\]
Then, using the identity (\ref{IntegrationByParts}) and the fact that, by assumption, $\mathcal{P}^*(N,N^i)=0$, we see that 
\begin{equation}
\int_\mathcal{S} \mathcal{P}\circ\mathcal{P}^*\left(\begin{array}{c}
\lambda\vert x\vert\\
0
\end{array}\right)\cdot \left(\begin{array}{c}
N\\
N_i
\end{array}\right)d\mu=\oint_{\partial\mathcal{S}_\infty}n^k(\mathcal{A}_k+\mathcal{B}_k+\mathcal{C}_k+\mathcal{D}_k)dS \label{IntegrationByPartsForKillingVector}
\end{equation}
where, here
\begin{eqnarray*}
&& \mathcal{A}_k \equiv N D^j \Gamma_{jk} - D^j N \Gamma_{jk} - D_k N
\Gamma -N D_k \Gamma,\\
&& \mathcal{B}_k \equiv 2(K^{ij}Q_{kij}-K Q_{kj}{}^j)N, \\
&& \mathcal{C}_k \equiv N^i D^l Q_{lik} - D^j N^i Q_{kij} + D_i N^i Q_{kj}{}^j -
   N_i D^l Q_{lj}{}^j, \\
&& \mathcal{D}_k \equiv N^i K_i{}^l \Gamma_{kl} -\tfrac{1}{2}N_k K^{lj}
   \Gamma_{jl} - \tfrac{1}{2}N^iK_{ik}\Gamma. 
\end{eqnarray*}
and $\Gamma\equiv h^{ij}\Gamma_{ij}$. We find then that
\[\mathcal{B}_k=o(\vert x\vert^{-3}),\qquad \mathcal{C}_k=o(\vert x\vert^{-5/2}),\qquad \mathcal{D}_k=o(\vert x\vert^{-5/2})\]
and 
\[\mathcal{A}_k=-4\lambda\mathfrak{A}^0\vert x\vert^{-2} n_k+o(\vert x\vert^{-5/2})\]
Therefore, the only contribution to the right-hand-side of (\ref{IntegrationByPartsForKillingVector}) is the following
\[\oint_{\partial\mathcal{S}_\infty}n^k\mathcal{A}_k dS=-4\lambda\mathfrak{A}^0\oint_{\partial\mathcal{S}_\infty}\vert x\vert^{-2} dS=-16\pi\lambda\mathfrak{A}^0 \]
  Since $\mathfrak{A}^0\neq 0$, we see that in the stationary case we
  have an obstruction to solving (\ref{InhomogeneousEquation}), unless
  $\lambda=0$, in which case $(\vartheta,X_i)=(N,N_i)$ is the unique
  solution, up to constant rescaling. 
  
  Finally, in the non-stationary case, uniqueness of the solution
$(X,X^i)$ follows by an argument analogous to that of Lemma
4. Suppose, for contradiction, that we have two linearly-independent
solutions $(\lambda|x|+\vartheta,X^i)$,
$(\tilde{\lambda}|x|+\tilde{\vartheta},\tilde{X}^i)$. Since the initial
data is by assumption non-stationary, then $\lambda\neq 0$ and
$\tilde{\lambda}\neq 0$. Hence, taking the appropriate linear
combination we arrive at a non-trivial approximate KID with lapse and
shift in $H_{1/2}^\infty$. This contradicts the conclusions of Lemma
4. Hence we conclude that $(\lambda|x|+\vartheta,X^i)$ and
$(\tilde{\lambda}|x|+\tilde{\vartheta},\tilde{X}^i)$ are linearly
dependent ---i.e. the solution is unique up to constant rescaling.
  

\end{proof}
\begin{remark}
{\em The fact that $(\vartheta,N^i)\in H^\infty_{1/2}$ in the previous theorem follows from an application of Proposition \ref{Prop:EllipticRegularity} to equation \eqref{InhomogeneousEquation}. }
\end{remark}
\begin{remark}
{\em In \cite{Dai04c}, the invariant $\lambda$ is also given as a bulk integral as follows
\[\lambda=\frac{1}{16\pi}\int_{\mathcal{S}}X r_{ij}r^{ij}d\mu\]
The above integral formula, valid only in the time-symmetric case, is derived from the following boundary integral
\[\lambda=-\frac{1}{8\pi}\oint_{\partial\mathcal{S}_\infty}n^k D_k\Delta\left(\lambda\vert x\vert\right)dS, \]
through the use of the divergence theorem, substitution using the approximate KID equation and integration-by-parts. A similar calculation yields a more general formula, valid for all data sets satisfying Decay Assumption \ref{Decay:Assumptions}, in terms of both the lapse and the shift of the approximate KID set. The expression is however rather complicated, and so in the interest of conciseness it is not presented here. Nevertheless, it reduces to the above formula when time-symmetry is assumed. Further study of the integral formula is deferred to subsequent work. 
}
\end{remark}

\section{The geometric invariant in conformally flat initial data sets}
\label{Section:FurtherAnalysis}

We have seen in the previous section that an approximate Killing
vector with lapse of the form $\eta=\lambda \vert x\vert+\vartheta$
exists for general asymptotically flat data, and moreover, that the
constant $\lambda$ vanishes if and only if the spacetime development
is stationary. In this section we analyse further the asymptotic
properties of the solutions to the approximate Killing vector equation
in the case of conformally flat initial data sets.  

\subsection{Solutions to the Poisson equation in $\mathbb{R}^3$}
We start with some mathematical preliminaries. Let us assume for the
remainder of this section that $K_{ij}=O(\vert x\vert^{-3+\epsilon})$,
for any $\epsilon>0$. It follows then from the Hamiltonian constraint
that 
\[
r=-K^2+K_{ij}K^{ij}=O(\vert x\vert^{-6+2\epsilon}).
\]
Moreover, the lapse component of the approximate Killing vector
equation can be found to satisfy
\begin{equation}
2\Delta_\bmh\Delta_\bmh\eta-r^{ij}D_i D_j\eta+r_{ij}r^{ij}\eta=O(\vert x\vert^{-11/2+\epsilon}).\label{EqLapseForConformallyFlatData}
\end{equation}
As is well known, the harmonic functions on $\mathbb{R}^3$ are spanned
by functions of the forms
\[
Q_{\alpha_1\cdots \alpha_k}x^{\alpha_1}\cdots x^{\alpha_k},\qquad \frac{Q_{\alpha_1\cdots \alpha_k}x^{\alpha_1}\cdots x^{\alpha_k}}{\vert x\vert^{2k+1}},\qquad\qquad k=0,1,2,\ldots,
\]
where $Q_{\alpha_1\cdots \alpha_k}$ are symmetric trace-free tensors
with constant coefficients. The following result will prove
useful:

\begin{lemma}[\textbf{\em Meyers, \cite{Mey63}}]
Let $\delta$ denote the flat 3--metric and $G=O(\vert x\vert^{-2-p-\epsilon}(\ln\vert x\vert)^q)$ a H\"{o}lder continuous function. Then the equation
\begin{equation}
\Delta_\delta V=G \label{InhomogenoeusLaplaceEq}
\end{equation}
admits  a solution $V^\star$ satisfying
\[V^\star(x)=\begin{cases}
O(\vert x\vert^{-p-\epsilon}(\ln\vert x\vert)^q) \qquad \text{if}~~ 0<p<1~~\text{or}~~ \epsilon>0,\\[1em]
O(\vert x\vert^{-p}(\ln\vert x\vert)^{q+1}) \qquad \text{otherwise.}
\end{cases}\]
\label{LemmaMeyers}
\end{lemma}
\begin{remark}
{\em By linearity of the Poisson equation
  (\ref{InhomogeneousEquation}), any two solutions thereof differ only
  by harmonic terms. In particular, the most general solution $V(x)$ of (\ref{InhomogenoeusLaplaceEq}), assuming  $V=O(\vert x\vert^{r})$ for $r>-p$, is given by 
\[
V(x)=\begin{cases}
V^\star(x)+\sum_{k=\lfloor -r\rfloor}^{\lfloor p-1\rfloor}Q_{\alpha_1\cdots \alpha_k}\displaystyle\frac{x^{\alpha_1}\cdots x^{\alpha_k}}{\vert x\vert^{2k+1}} \qquad \text{if}\quad r<0, \\[1em]
V^\star(x)+\sum_{k=0}^{\lfloor p-1\rfloor}Q_{\alpha_1\cdots \alpha_k}\displaystyle\frac{x^{\alpha_1}\cdots x^{\alpha_k}}{\vert x\vert^{2k+1}}+\sum_{l=0}^{\lfloor r\rfloor}\widehat{Q}_{\alpha_1\cdots\alpha_l}x^{\alpha_1}\cdots x^{\alpha_l} \qquad \text{if}\quad r>0.
\end{cases}\]
for some symmetric, trace-free
$Q_{\alpha_1\cdots\alpha_k},~\widehat{Q}_{\alpha_1\cdots\alpha_l}$
with constant coefficients and where for a real number $p$, $\lfloor
p\rfloor$ denotes the floor of $p$ ---i.e. the largest integer smaller than $p$. It will be useful to note that, for $k\in\mathbb{Z}$,
\[
\Delta_\bmdelta\left(\frac{x^\alpha}{\vert
    x\vert^k}\right)=k(k-3)\frac{x^\alpha}{\vert x\vert^{k+2}}.
\]}
\end{remark}

\subsection{Conformally flat initial data sets}

We consider now maximal conformally-flat data initial data sets,
i.e. collections $(\mathcal{S},h_{ij},K_{ij})$ such that
\[
h_{ij} = \phi^4 \delta_{ij}, \qquad K_{ij}=P_{ij}
\]
where $\phi\rightarrow 1$ as $\vert x\vert \rightarrow\infty$ and
$P_{\alpha\beta}=O(\vert x\vert^{-3+\epsilon})$ is a symmetric, tracefree and
divergence free with respect to the flat metric. It will also prove
convenient to define $\psi_{\alpha\beta}\equiv \phi^2
P_{\alpha\beta}$, in terms of which the Hamiltonian and momentum
constraints take the familiar forms
\begin{subequations}
\begin{eqnarray}
&& \Delta_{\bmdelta}\phi=-\tfrac{1}{8}\phi^{-7}\psi_{\alpha\beta}\psi^{\alpha\beta},\label{LicnerowiczEq}\\
&& \partial^\alpha \psi_{\alpha\beta}=0\label{ConfMomentumConstraint},
\end{eqnarray}
\end{subequations}
where indices are now raised and lowered with respect to the flat
metric, $\delta_{ij}$.  
Then, it follows from (\ref{LicnerowiczEq})
and an application of Lemma \ref{LemmaMeyers} that
\begin{equation}
\phi=1+\frac{2m}{\vert x\vert}+\frac{L_{\alpha}x^\alpha}{\vert x\vert^3}+\frac{A_{\alpha\beta}x^\alpha x^\beta}{\vert x\vert^5}+O\left(\frac{\ln\vert x\vert}{\vert x\vert^{4-2\epsilon}}\right) \label{ConfFactorConformallyFlat}
\end{equation}
for some constant $m$, and constant-coefficient $L_\alpha$, $A_{\alpha\beta}$, which are
independent of the extrinsic curvature $P_{\alpha\beta}$ which contributes only at
order ~$O(\ln\vert x\vert/\vert x\vert^{4-2\epsilon})$.
\medskip

In terms of the flat connection, equation (\ref{EqLapseForConformallyFlatData}) becomes
\begin{align}
&\Delta_\bmdelta\Delta_\bmdelta\eta+A(\phi)^\alpha\partial_\alpha\Delta_\bmdelta\eta+B(\phi)^{\alpha\beta}\partial_\alpha\partial_\beta\eta+B(\phi)\Delta_\bmdelta\eta\nonumber \\
&\qquad\qquad\qquad\qquad\qquad\qquad\qquad+C(\phi)^\alpha\partial_\alpha\eta+D(\phi)\eta=O(\vert x\vert^{-11/2+\epsilon})\label{EqLapseForConformallyFlatDataFlatConnection}
\end{align}
where 
\begin{align*}
& A(\phi)^\alpha \equiv-4\phi^{-1}\partial^\alpha\phi,\\
& B(\phi)^{\alpha\beta} \equiv\phi^{-2}(5\phi\partial^\alpha\partial^\beta\phi-19\partial^\alpha\phi\partial^\beta\phi),\\
& B(\phi) \equiv 13\phi^{-2}\vert\partial\phi\vert^2,\\
& C(\phi)^\alpha\equiv 4\phi^{-3}(12\vert\partial\phi\vert^2\partial^\alpha\phi-5\phi(\partial_\beta\phi)\partial^\beta\partial^\alpha\phi),\\
& D(\phi) \equiv 2\phi^{-4}(6\vert\partial\phi\vert^4-6\phi(\partial^\alpha\phi)(\partial^\beta\phi)\partial_\alpha\partial_\beta\phi+\phi^2(\partial^\alpha\partial^\beta\phi)(\partial_\alpha\partial_\beta\phi)).
\end{align*}

\smallskip
\begin{proposition}
Let $(\mathcal{S},h_{ij},P_{ij})$ be maximal conformally-flat data
with $P_{\alpha\beta}=O(\vert x\vert^{-3+\epsilon})$, then the lapse
of the approximate Killing vector has an asymptotic expansion of the
form
\begin{align}
\eta&=\lambda \vert x\vert+18\lambda m\ln{\vert x\vert}+Q_\alpha\frac{x^\alpha}{\vert x\vert}+Q^{(1)}-104\lambda m^2 \frac{\ln{\vert x\vert}}{\vert x\vert}\nonumber\\
&\qquad\qquad\qquad+\frac{Q^{(2)}}{\vert x\vert}-\frac{1}{4}(23\lambda L_\alpha-26mQ_\alpha)\frac{x^\alpha}{\vert x\vert^2}+Q_{\alpha\beta}\frac{x^\alpha x^\beta}{\vert x\vert^3}+O(\vert x\vert^{-3/2}) \label{AsymptoticExpansionEta}
\end{align}
for some constants $Q^{(1)},Q^{(2)},Q_\alpha,Q_{\alpha\beta}$, where $m$ and $L_\alpha$ are the constants appearing in (\ref{ConfFactorConformallyFlat}).
\end{proposition}
\begin{proof}
Substituting (\ref{ConfFactorConformallyFlat}) into
(\ref{EqLapseForConformallyFlatDataFlatConnection}) one obtains 
\begin{align}
& \Delta_\bmdelta\Delta_\bmdelta\eta+\frac{4}{\vert x\vert^3}\left[2mx^\alpha-4m^2\frac{x^\alpha}{\vert x\vert}-L^\alpha+3L_\beta\frac{x^\beta x^\alpha}{\vert x\vert^2}\right]\partial_\alpha\Delta_\bmdelta\eta \nonumber \\
&\qquad\quad+\frac{15}{\vert x\vert^5}\left[\left(2m-\frac{4m^2}{\vert x\vert}+\frac{5}{\vert x\vert^2}L_\gamma x^\gamma\right)x^\alpha x^\beta-2L^{(\alpha}x^{\beta)}\right]\partial_\alpha\partial_\beta\eta \nonumber\\
&\qquad\quad-\frac{1}{\vert x\vert^3}\left[10m-\frac{72m^2}{\vert x\vert}+\frac{15}{\vert x\vert^2}L_\gamma x^\gamma\right]\Delta_\bmdelta\eta+\frac{160m^2}{\vert x\vert^6}x^\alpha\partial_\alpha\eta+\frac{48m^2}{\vert x\vert^6}\eta =O(\vert x\vert^{-11/2+\epsilon}).\label{EqLapseForConformallyFlatDataFlatConnectionExpanded}
\end{align}
Substituting the Ansatz, $\eta=\lambda \vert x\vert+\vartheta$, where
$\vartheta=o(\vert x\vert^{1/2})$, and collecting lower-order terms in
$\vartheta$
\[
\Delta_\bmdelta\Delta_\bmdelta\vartheta=\frac{36\lambda m}{\vert
  x\vert^4}+O(\vert x\vert^{-9/2}).
\]
Using Lemma \ref{LemmaMeyers}, we obtain
\[
\Delta_\bmdelta\vartheta=\frac{18\lambda m}{\vert
  x\vert^2}-2Q_\alpha\frac{x^\alpha}{\vert x\vert^3}+O(\vert
x\vert^{-5/2})
\]
for some constant-coefficient $Q_\alpha$. Here we have used that
$\Delta_\delta\vartheta=o(\vert x\vert^{-3/2})$, thereby excluding
constant and $1/\vert x\vert$ harmonic terms. Applying Lemma
\ref{LemmaMeyers} again we find that 
\[
\vartheta=18\lambda m\ln{\vert x\vert}+Q_\alpha\frac{x^\alpha}{\vert
  x\vert}+Q^{(1)}+\varphi
\]
for some constant $Q^{(1)}$, and some function $\varphi=O(\vert x\vert^{-1/2})$. Substituting into (\ref{EqLapseForConformallyFlatDataFlatConnectionExpanded}), 
\[
\Delta_\bmdelta\Delta_\bmdelta\varphi=\frac{624 \lambda m^2}{\vert
  x\vert^5}+(46\lambda L_\alpha-52mQ_\alpha)\frac{x^\alpha}{\vert
  x\vert^6}+O(\vert x\vert^{-11/2+\epsilon})
\]
implying that
\[
\Delta_\bmdelta\varphi=\frac{104\lambda m^2}{\vert
  x\vert^3}+\frac{1}{2}(23\lambda
L_\alpha-26mQ_\alpha)\frac{x^\alpha}{\vert
  x\vert^4}+Q_{\alpha\beta}\frac{x^\alpha x^\beta}{\vert
  x\vert^5}+O(\vert x\vert^{-7/2+\epsilon})
\]
for some constant-coefficient, tracefree $Q_{\alpha\beta}$. Here we have used the fact that $\Delta_\bmdelta\varphi=O(\vert x\vert^{-5/2})$ to eliminate constant, $1/\vert x\vert$ and $1/\vert x\vert^2$ harmonic terms. Integrating up once more, we obtain (\ref{AsymptoticExpansionEta}).
\end{proof}

It is interesting to note the presence of a logarithmically-singular term in (\ref{AsymptoticExpansionEta}) in the non-Killing case, $\lambda\neq 0$. On the other hand, if one sets $\lambda=0$ in (\ref{AsymptoticExpansionEta}), then one obtains the asymptotic expansion 
\begin{align}
\eta&=Q_\alpha\frac{x^\alpha}{\vert x\vert}+Q^{(1)}+\frac{Q^{(2)}}{\vert x\vert}+\frac{13}{2}mQ_\alpha\frac{x^\alpha}{\vert x\vert^2}+Q_{\alpha\beta}\frac{x^\alpha x^\beta}{\vert x\vert^3}+O(\vert x\vert^{-3/2}) \label{AsymptoticExpansionEtaKilling}
\end{align}
for the lapse of a general spacetime Killing vector restricted to a conformally flat initial data set.


\section*{Conclusions}
We have shown that the existence of approximate Killing vectors
extends to a large class of asymptotically Euclidean initial data with
non-vanishing extrinsic curvature and non-vanishing ADM 4-momentum. Following Dain's discussion in
\cite{Dai04c}, we can then define a geometric invariant $\lambda_{(k)}$
on each asymptotically Euclidean end given by the leading
coefficient in the appropriate asymptotic expansion. The vanishing of
any one of the $\lambda_{(k)}$ characterises stationarity of the
initial data. 

Further work would involve the construction of
approximate Killing vectors on hyperboloidal hypersurfaces. It would also be of interest to explore the dynamics of the approximate KID. If a propagation equation for the approximate KID can be found, one may be able to use Dain's invariant to quantify the deviation from stationarity of a generic asymptotically Euclidean initial data set.

\section*{Acknowledgements}
This article is dedicated to the memory of Sergio Dain (1969-2016) who
played an important role in shaping the knowledge of the first author
not only on the topic of this article but also on many other aspects of General
Relativity.  We would like to thank L.B. Szabados for clarifying the
issue of units in the construction of the approximate KID equation. We would also
like to thank P.T. Chru\'sciel and R. Beig for helpful comments
on an earlier draft.



\end{document}